\documentclass[runningheads]{llncs}

\usepackage[T1]{fontenc}
\usepackage{amsmath}            %
\usepackage{newtxtext}
\usepackage[varvw]{newtxmath}
\usepackage{microtype}

\usepackage{graphicx}
\graphicspath{{figures/}}
\usepackage{booktabs}
\usepackage{array}
\usepackage{tikz}
\usetikzlibrary{arrows.meta,positioning}

\usepackage{xcolor}
\usepackage[hidelinks]{hyperref}

\hypersetup{
  pdftitle  = {Defining Purpose-Limited Secrets},
  pdfauthor = {Bhumika Mittal, Aalok Thakkar},
  pdfsubject = {Cryptographic purpose limitation and purpose-limited secrets},
  pdfkeywords = {Purpose limitation, Confinement, Functional cryptography, E-cash, Capabilities, Data protection},
}

\newcommand{\algo}[1]{\mathsf{#1}}
\newcommand{\Op}{\algo{Op}}
\newcommand{\Succ}{\algo{Succ}}
\newcommand{\adm}{\algo{adm}}
\newcommand{\Setup}{\algo{Setup}}
\newcommand{\Constrain}{\algo{Constrain}}
\newcommand{\Use}{\algo{Use}}
\newcommand{\Attenuate}{\algo{Attenuate}}
\newcommand{\Make}{\algo{Make}}
\newcommand{\Trace}{\algo{Trace}}
\newcommand{\Sign}{\algo{Sign}}
\newcommand{\Verify}{\algo{Verify}}
\newcommand{\Kg}{\algo{Kg}}
\newcommand{\Sg}{\algo{Sg}}
\newcommand{\Vf}{\algo{Vf}}
\newcommand{\redeem}{\algo{redeem}}
\newcommand{\range}{\algo{range}}
\newcommand{\negl}{\algo{negl}}
\newcommand{\poly}{\mathrm{poly}}
\newcommand{\ch}{\algo{ch}}

\newcommand{\pp}{\mathit{pp}}
\newcommand{\msk}{\mathit{msk}}
\newcommand{\sk}{\mathit{sk}}
\newcommand{\vk}{\mathit{vk}}
\newcommand{\bsk}{\mathit{bsk}}
\newcommand{\bvk}{\mathit{bvk}}
\newcommand{\ask}{\mathit{ask}}
\newcommand{\avk}{\mathit{avk}}
\newcommand{\id}{\mathit{id}}

\newcommand{\cA}{\mathcal{A}}
\newcommand{\cB}{\mathcal{B}}
\newcommand{\cP}{\mathcal{P}}
\newcommand{\G}{\mathbb{G}}
\newcommand{\Zq}{\mathbb{Z}_q}
\newcommand{\Sdel}{\Sigma_{\mathsf{del}}}

\newcommand{\plgame}[1]{\textsf{PL-#1}}
\newcommand{\oracle}[1]{O_{\mathsf{#1}}}
\newcommand{\Adv}[2][]{\mathrm{Adv}^{\mathsf{#2}}_{#1}}

\newcommand{\RefSem}{\algo{Ref}}
\newcommand{\Iss}{\algo{Iss}}
\newcommand{\Hld}{\algo{Hld}}
\newcommand{\AKg}{\algo{AKg}}
\newcommand{\Sim}{\algo{Sim}}
\newcommand{\Chal}{\algo{Chal}}
\newcommand{\Leak}{\algo{Leak}}
\newcommand{\Acc}{\algo{Accept}}
\newcommand{\iid}{\mathit{iid}}
\newcommand{\Kw}[1]{\textbf{#1}}

\newsavebox{\gamesavebox}
\newenvironment{gamebox}[1]
  {\begin{lrbox}{\gamesavebox}\begin{minipage}[t]{#1}\small\raggedright}
  {\end{minipage}\end{lrbox}\fbox{\usebox{\gamesavebox}}}
\newcommand{\gametitle}[1]{\underline{#1}\par\smallskip}
\newcommand{\sendright}[1]{\begin{tabular}[b]{@{}c@{}}#1\\[-1.4ex]\makebox[3.1cm]{\rightarrowfill}\end{tabular}}
\newcommand{\sendleft}[1]{\begin{tabular}[b]{@{}c@{}}#1\\[-1.4ex]\makebox[3.1cm]{\leftarrowfill}\end{tabular}}

\begin{document}

\title{Defining Purpose-Limited Secrets\thanks{Full version of a paper accepted at the 22nd International Conference on Information Systems Security (ICISS 2026), to appear in Springer LNCS.}}

\author{\begin{tabular}[t]{c@{\hspace{3em}}c}
Bhumika Mittal & Aalok Thakkar\\
{\normalfont\small School of Computer Science} & {\normalfont\small Vachani School of Advanced Computing}\\
{\normalfont\small Georgia Institute of Technology} & {\normalfont\small Ashoka University}
\end{tabular}}
\authorrunning{B. Mittal and A. Thakkar}
\institute{}

\maketitle

\begin{abstract}
A cryptographic secret is issued for a purpose but grants a capability, and the capability is usually larger. Deployed systems state the purpose in policy and enforce only the capability. We make the purpose a property of the secret. A mechanism sees operations, not reasons, so a scheme enforces an \emph{admissible set} $A(P)$ of operations that stands in for a declared \emph{intended use}~$I$; whether $I$ captures the human purpose is a modeling obligation that marks where policy takes over. Against the same~$I$ we define three regimes: under \emph{confinement} an operation outside the intended use is infeasible, under \emph{mediation} a trusted component refuses it, and under \emph{accountability} use beyond a stated bound is possible but attributed to the holder. The confinement games, for unpredictability and indistinguishability, coincide with the key-query forms of functional-signature unforgeability and constrained-PRF pseudorandomness and correspond to simulation-secure functional encryption within its feasibility boundary; stated against~$I$, they also register a predicate that permits too much. A payment token uses all three regimes on one secret: it pays one capped charge at one merchant, can be narrowed but not widened, is checked by the merchant, and identifies the withdrawing account if spent twice. Attenuation unforgeability reduces to unforgeable signatures and a collision-resistant hash without random oracles; traceability and non-frameability hold under discrete log and unforgeable signatures in the random-oracle model. We present this framework as a common specification against which existing primitives are stated and compared.

\end{abstract}

\section{Introduction}\label{sec:intro}

A person who asks an assistant to pay one invoice must hand over a payment credential, and the credential can empty the account. The same gap appears wherever a secret is shared for a task. A hospital gives an analytics firm a decryption key to compute infection rates, and the key decrypts every patient record. A team gives its build server a signing key for releases of one product, and the key signs anything, malware included. Each secret is given for a purpose but grants a capability, and the capability is larger.

Deployed systems write the purpose down and enforce it outside the key. An X.509 certificate lists the uses its key is meant for~\cite{rfc5280}, and a PKCS\#11 token tags each key with the operations it may perform~\cite{pkcs11v3}. In both cases the key can perform any operation, and only a trusted component that reads the tag stands in the way; bypassing it restores the full capability~\cite{clulow2003pkcs11,bortolozzo2010pkcs11}. This paper asks whether ``this key works only for purpose~$P$'' can be proved of the key itself, as one proves that a signature is unforgeable, so that the limit holds when no monitor is present. Where that is out of reach, it asks what a trusted checker or an audit log can guarantee instead.

Such a proof can concern only what a mechanism observes. The same \$50 paid to the same merchant may settle the invoice or pay a bribe, and a decryption may serve analytics or re-identification depending on what happens to the plaintext afterwards. A scheme therefore limits operations, not reasons. Each purpose~$P$ comes with an \emph{admissible set} $A(P)$ of operations that the mechanism permits, and the issuer declares an \emph{intended use}~$I$, which says which operations each purpose is meant to cover. ``Usable only for purpose~$P$'' is shorthand for ``usable only for the operations in~$A(P)$''. Whether $A$ and $I$ capture the human purpose is a modeling obligation, not a theorem, and it marks where cryptography hands off to policy.

\paragraph{Three regimes.}
Against one intended use we define three regimes, which give different kinds of guarantee. Under \emph{confinement} an operation outside $I(P)$ is computationally infeasible, a hardness property. Under \emph{mediation} a trusted monitor or verifier refuses it, a trust assumption. Under \emph{accountability} use beyond a stated bound is possible, and a trace attributes it to the holder after the fact. Functional encryption and constrained pseudorandom functions (PRFs) confine, hardware key stores and certificate flags mediate, traitor tracing holds to account, and electronic cash confines and holds to account.

\paragraph{A payment token.}
For the invoice task we build a token whose single secret uses all three regimes. It authorizes one charge, up to a cap, at one merchant, within a time window. A holder can narrow it but not widen it (confinement), the merchant checks each charge (mediation), and a second spend reveals the account that withdrew it (accountability).

\paragraph{What is new.}
We claim no new primitive or hardness assumption: the confinement games coincide with existing notions in key-query form, and the token is built from known parts (Section~\ref{sec:token}). The paper contributes a common specification, in which every guarantee, in every regime, is stated against the declared intended use~$I$ and not against the predicate a primitive enforces. A predicate that permits too much can therefore fail the specification while it satisfies the primitive's own game (Section~\ref{sec:signing}), and the regimes can be related to one another (Section~\ref{sec:mediation-accountability}). The token's results are proved one regime at a time; a joint security statement and a general composition theorem are open (Section~\ref{sec:discussion}).

\paragraph{Contributions.}
The paper contributes the following.
\begin{itemize}
  \item A specification of purpose limitation (Section~\ref{sec:define}): syntax, the intended use and its proxy, and games for the three regimes. The confinement games coincide with key-query unforgeability of functional signatures and with constrained-PRF pseudorandomness against constrained-key queries, cover policy-based signatures, and correspond to functional encryption within its feasibility boundary.
  \item The payment token (Section~\ref{sec:token}), with proofs of attenuation unforgeability, traceability, and non-frameability, a comparison with macaroons, SPKI/SDSI, e-cash, and delegatable credentials by the properties each achieves, and an anonymous variant.
  \item A discussion of what the common specification adds and of where cryptographic purpose limitation hands off to policy (Section~\ref{sec:discussion}).
\end{itemize}
Section~\ref{sec:limit} derives the requirements from the payment task, and Section~\ref{sec:related} discusses related work.

\section{Requirements for a Purpose-Limited Secret}\label{sec:limit}

The invoice task of Section~\ref{sec:intro} is the running example of the paper. A user wants an assistant, human or software, to pay one invoice of at most \$50 to one named merchant within the current week. The user's payment credential can do far more, and the aim is a smaller secret that performs this task and nothing else. Section~\ref{sec:token} constructs a token for it.

\paragraph{Requirements.}
The secret must not authorize a payment the task excludes: a different merchant, a larger amount, or a later date. Delegation must only shrink authority: if the assistant sub-contracts part of the task, it derives and passes on a narrower secret (an operation we call \emph{attenuation}), say at most \$20 to the same merchant, and the recipient cannot enlarge it back to \$50 or extend it to a second merchant. Finally, the secret should work once.

\paragraph{Enforcement.}
The first two can rest on hardness: producing a valid payment outside the granted terms, or a delegation chain that widens them, is infeasible, and validity is a public check (confinement). The merchant's check of amount, identity, and date at redemption is mediation, as strong as the trust placed in the checker~\cite{bortolozzo2010pkcs11}. The third requirement needs accountability. A merchant that keeps no record cannot tell a second presentation of the token from the first, but a scheme can make the second use reveal the holder, as in electronic cash~\cite{chaum1988ecash}.

\paragraph{Scope of the guarantee.}
The token limits the merchant, the amount, and the date, which are visible when it is used; it cannot limit the reason for the payment. Purpose in the legal sense, as in Article~5(1)(b) of the EU General Data Protection Regulation (GDPR)~\cite{oecd1980guidelines,gdpr2016,wp29-2013purpose}, concerns why data is used and what happens downstream, which a mechanism cannot see. The guarantee is therefore stated for the observable operation ``pay amount $a$ to merchant $M'$ at time~$t$'', the position that formal treatments of purpose in audit and contextual integrity also take.

\paragraph{Further examples.}
The analytics and build-key examples of Section~\ref{sec:intro}, and a dataset licensed to train one model, motivate the definition but are not instances of the token. An analytics key that computes aggregates but cannot read a record is confinement of the indistinguishability kind, realized by functional encryption within its known limits (Section~\ref{sec:confinement}). A build key that signs one project's artifacts is confinement of the unpredictability kind, realized by policy-based signatures (Section~\ref{sec:signing}). A dataset licensed to train one model reaches the release of plaintext, which no confinement game governs (Section~\ref{sec:limits}).

\section{Defining Purpose Limitation}\label{sec:define}

A secret is purpose-limited when use outside its intended use is ruled out, refused, or attributed, depending on the regime. In this section, we specify what each purpose permits and is meant to permit, and give one game per regime.

\subsection{Purposes, Operations, and the Intended Use}\label{sec:purposes}

\begin{figure}[t]
  \centering
  \begin{tikzpicture}[font=\small, every node/.style={align=center}]
  \draw[dashed] (-3.9,-1.8) rectangle (3.9,1.8);
  \node at (0,1.45) {human purpose (unobservable)};

  \draw[thick] (-1.1,-0.4) ellipse [x radius=1.9, y radius=0.95];
  \draw[thick] ( 1.1,-0.4) ellipse [x radius=1.9, y radius=0.95];

  \node at (-2.2,0.85) {$I(P)$};
  \node at ( 2.2,0.85) {$A(P)$};

  \node[font=\footnotesize] at (-1.95,-0.4) {false deny\\$I\setminus A$};
  \node at ( 0,   -0.4) {$A\cap I$};
  \node[font=\footnotesize] at ( 1.95,-0.4) {false permit\\$A\setminus I$};
\end{tikzpicture}
  \Description{Diagram inside a dashed box labelled "human purpose (unobservable)". Two overlapping ellipses: the left one is I(P), the declared intended use; the right one is A(P), the admissible set the mechanism enforces. The part of I(P) outside A(P) is labelled "false deny, I minus A"; the overlap is labelled "A intersect I"; the part of A(P) outside I(P) is labelled "false permit, A minus I".}
  \caption{Fidelity of the proxy $A(P)$ to the intended use $I(P)$; the human purpose (dashed) is not observable. Soundness excludes $A\setminus I$, completeness $I\setminus A$. Games are stated against $I$, so a false permit an adversary can find is a break.}
  \label{fig:fidelity}
\end{figure}

A capability fixes what an unrestricted secret does and how one tells that it was done.

\begin{definition}[Capability]\label{def:capabilities}
A \emph{capability} is a triple $(\Op,\RefSem,\Succ)$: an \emph{operation family} $\Op$; a \emph{reference semantics} $\RefSem$, giving the correct result $\RefSem(o)$ of an unrestricted secret on~$o\in\Op$; and a \emph{success predicate} $\Succ$, with $\Succ(o,y)=1$ iff $y$ is a valid achievement of~$o$, and $\Succ(o,\RefSem(o))=1$. If results should be hard to produce, as for signatures, $\Succ$ is public.
\end{definition}

A public success predicate makes a forbidden achievement a checkable event, whoever checks it. If results should stay hidden, as for PRF values, $\Succ(o,y)=[\,y=\RefSem(o)\,]$ only states correctness and is not public, since a public test would tell a PRF value from a uniform one. Such capabilities are specified through the challenge and leakage of Game~\plgame{CONF-IND} (Definition~\ref{def:conf-ind}).

Purposes are ordered by how much they permit, and a predicate on a purpose and an operation decides what the mechanism permits.

\begin{definition}[Purposes and admissibility]\label{def:purposes}
Purposes form a bounded lattice $(\cP,\sqsubseteq,\sqcup,\sqcap,\top,\bot)$ under \emph{refinement}, where $P'\sqsubseteq P$ reads ``$P'$ is at most as permissive as~$P$'': any two purposes have a least upper bound $\sqcup$ and a greatest lower bound $\sqcap$, $\top$ is the most permissive purpose, and $\bot$ the least. An \emph{admissibility relation} $\adm:\cP\times\Op\to\{0,1\}$ fixes which operations are in-purpose, and the \emph{admissible set}, or \emph{proxy}, of~$P$ is
\[
  A(P) := \{\, o\in\Op : \adm(P,o)=1 \,\}.
\]
\end{definition}

\begin{definition}[Monotone attenuation]\label{def:monotone}
$\adm$ satisfies \emph{monotone attenuation} if $P'\sqsubseteq P \Rightarrow A(P')\subseteq A(P)$, with $A(\top)=\Op$ and $A(\bot)=\emptyset$.
\end{definition}

Monotonicity makes narrowing safe: a holder who passes on a narrower purpose passes on a subset of its admissible operations.

\paragraph{Example.}
In the payment task, an operation $o=\redeem(a,M',t)$ is a payment of amount $a$ to merchant $M'$ at time~$t$, and $\Succ(o,y)=1$ iff the merchant accepts the redemption transcript $y$ for~$o$. A purpose $P=(M,C,T)$ is a merchant, an amount cap, and a time window, and $\adm(P,o)=[\,M'\preceq M\wedge a\le C\wedge t\in T\,]$, where $M'\preceq M$ means that $M'$ is $M$ or one of its outlets. The invoice task is $P_0=(M,50,T)$ for this week's window~$T$, and the sub-contract is $(M,20,T)\sqsubseteq P_0$.

A mechanism enforces~$A$. What the issuer means is declared separately, as the intended use.

\begin{definition}[Intended use and fidelity]\label{def:fidelity}
A scheme is specified relative to an \emph{intended-use relation} $I\subseteq\cP\times\Op$, the target the proxy stands for; write $I(P):=\{o : (P,o)\in I\}$. The proxy $A$ is \emph{sound} for $I$ if $A(P)\subseteq I(P)$ for all $P$ (no false permit) and \emph{complete} if $A(P)\supseteq I(P)$ for all $P$ (no false deny).
\end{definition}

Figure~\ref{fig:fidelity} shows the regions. The proxy is \emph{exact} if $A=I$, as in the running example, where the issuer means ``pay at most $C$ to $M$ within~$T$'' and nothing else. Soundness is what a confinement claim needs when the proxy is not exact, and completeness, with correctness, lets the holder perform every intended operation. Because $I$ is declared, its fidelity to the human purpose is a hypothesis of every result, not a cryptographic property. An operation is \emph{off-purpose} for a set $Q$ of purposes if it lies outside $\bigcup_{P\in Q}I(P)$; the confinement and mediation games ask whether an adversary can achieve one.

\subsection{Syntax and Correctness}\label{sec:syntax}

The syntax is that of a constrained key, with optional narrowing and, for accountability, an identity that a trace recovers.

\begin{definition}[PLC scheme]\label{def:plc}
Fix a capability $(\Op,\RefSem,\Succ)$ and a purpose lattice~$\cP$. A \emph{purpose-limited capability scheme} (PLC) consists of the following probabilistic polynomial-time (PPT) algorithms and protocols, with security parameter~$\lambda$.
\begin{itemize}
  \item $\Setup(1^\lambda)\to(\pp,\msk)$ outputs public parameters and the issuer's master key.
  \item $\Constrain\langle\Iss(\msk),\Hld\rangle(P)\to\sk_P$ is a protocol between an issuer $\Iss$ holding $\msk$ and a holder~$\Hld$, whose output $\sk_P$ may contain holder secrets the issuer never sees. A non-interactive scheme is the case where the issuer sends one message.
  \item $\Use(\pp,\sk_P,o)\to y$ applies the secret to an operation and returns a result $y$, or $\bot$. It may be a protocol between the holder and a verifier, in which case $y$ is the transcript.
\end{itemize}
A scheme may also have:
\begin{itemize}
  \item $\Attenuate(\pp,\sk_P,P')\to\sk_{P'}$ for $P'\sqsubseteq P$, run by the holder to derive a secret for a narrower purpose;
  \item $\Make(\pp,x)\to\omega$, a public algorithm creating the objects operations act on; an operation $o$ may then refer to $\Make(\pp,x)$, and $\Use$ acts on it.
\end{itemize}
The \emph{accountable variant} adds:
\begin{itemize}
  \item $\AKg(1^\lambda)\to(\id,\ask)$, account key generation registering a public identity $\id$ with secret account key~$\ask$;
  \item $\Constrain$ additionally takes $\id$ and, using $\ask$, records the issuance in an issuer-held store~$R$;
  \item $\Trace(\pp,\iid,L,R)\to\id$ or $\bot$, which, for an issued instance~$\iid$ (a scheme-defined identifier of one issued secret), recovers a per-issuance handle from the uses of $\iid$ in a log $L$ of accepted uses and resolves it to $\id$ through~$R$; we write ``$\Trace$ outputs $\id$'' for this composite.
\end{itemize}
\end{definition}

$\Make$ lets encryption fit the syntax: for functional encryption, $\Constrain$ is key generation for~$f$ and $\Make$ is encryption, so the operation ``learn $f(x)$'' acts on $\Make(\pp,x)$. Purpose may instead be bound at encryption time, the dual syntax of ciphertext-policy schemes; the token of Section~\ref{sec:token} binds purpose to the key and supports attenuation.

\paragraph{Example.}
For the payment token, $\Setup$ generates the bank's keys, $\Constrain$ is withdrawal, $\Use$ is redemption with the merchant, $\Attenuate$ appends a delegation certificate, $\AKg$ is account registration, and $\Trace$ is double-spend tracing over the bank's deposit log.

\begin{definition}[Correctness]\label{def:correctness}
For all $(\pp,\msk)$ output by $\Setup$, all $P$, all $\sk_P$ output by $\Constrain$ for~$P$, and all $o$ with $\adm(P,o)=1$,
\[
  \Succ\bigl(o,\ \Use(\pp,\sk_P,o)\bigr)=1,
\]
and $\Attenuate$ preserves this on $A(P')$.
\end{definition}

Correctness says nothing about limitation, since an all-powerful key is also correct. The three regimes capture limitation.

\subsection{Confinement}\label{sec:confinement}

\begin{figure}[t]
  \centering
  \input{figures/games-conf}
  \Description{Two framed boxes of game pseudocode. Left, Game PL-CONF-UP: the adversary gets the public parameters and a key oracle O_Constrain that records each queried purpose in Q, and wins by outputting an operation and a result that the success predicate accepts, with the operation outside the intended use of every purpose in Q. Right, Game PL-CONF-IND: after key queries the adversary names an input x and receives either the real objects for x or a simulator's output computed from the queried purposes, their keys, and the entitled leakage; it wins by guessing which.}
  \caption{The confinement games. In $\plgame{CONF-UP}$ (left) the adversary wins by achieving an off-purpose operation; in $\plgame{CONF-IND}$ (right) it must tell real objects for its input from objects simulated from its entitled leakage.}
  \label{fig:games-conf}
\end{figure}

Under confinement an off-purpose operation is computationally infeasible, and no party is trusted. ``Beyond purpose'' may mean producing a forbidden result, which fits capabilities whose output should be hard to produce (signatures, tokens, e-cash), or learning a forbidden thing, which fits capabilities whose output should stay hidden (encryption, PRFs, functional encryption). Game~\plgame{CONF-UP} covers the first case and Game~\plgame{CONF-IND} the second.

\begin{definition}[Game \plgame{CONF-UP}]\label{def:conf-up}
In Game \plgame{CONF-UP} (Figure~\ref{fig:games-conf}, left), $\cA$ adaptively obtains secrets $\sk_P$ for purposes $P$ of its choice, recorded in~$Q$, and wins if it outputs an achievement $y^*$ of an operation $o^*$ that no purpose in $Q$ intends. The scheme is \emph{unpredictability-confining} for $I$ if every PPT $\cA$ wins with probability $\negl(\lambda)$, a negligible function of~$\lambda$.
\end{definition}

The game is stated against $I$, not against~$A$. An unsound proxy is therefore insecure once its false-permit region can be found: an adversary that efficiently finds $P$ and $o\in A(P)\setminus I(P)$ obtains $\sk_P$, runs $\Use$ on~$o$, and wins. Under a sound proxy $\bigcup_{P\in Q}A(P)\subseteq\bigcup_{P\in Q}I(P)$, so confining the holder to the admissible sets confines it to the intended use. With one query ($|Q|=1$) the game is single-key confinement. The general form is \emph{collusion resistance}: a set of keys together achieves only the union of their intended sets.

\begin{figure}[t]
  \centering
  \input{figures/games-deleg-med}
  \Description{Two framed boxes of game pseudocode. Left, PL-CONF-UP with delegation: oracles grant, att, deleg, and use create a secret at an honest holder, attenuate it, hand it to the adversary (recording its purpose in Q), and run an honest use (recording the operation in T); the adversary wins with an accepted operation outside the intended use of Q and not in T. Right, Game PL-MED: oracle new creates a secret for a purpose; in held mode a monitor keeps it and runs only admissible operations, and the adversary wins with a valid result for an operation outside the intended use of every created purpose that no monitor returned; in gated mode the secret is returned, and the adversary wins if the public acceptance algorithm accepts such an operation.}
  \caption{Left: $\plgame{CONF-UP}$ with delegation; operations of honest uses are not wins. Right: $\plgame{MED}$, held and gated modes.}
  \label{fig:games-deleg-med}
\end{figure}

With $\Attenuate$, confinement must survive delegation: an adversary that watches honest holders narrow and use secrets, and receives some of them, should still achieve nothing off-purpose for what it received.

\begin{definition}[Game \plgame{CONF-UP} with delegation]\label{def:conf-up-deleg}
In this game (Figure~\ref{fig:games-deleg-med}, left) the challenger plays the issuer and all honest holders. $\oracle{grant}(P)$ creates a secret at an honest holder and returns a handle~$h$ with purpose~$P_h$; $\oracle{att}(h,P')$, for $P'\sqsubseteq P_h$, attenuates it at that holder and returns a fresh handle; $\oracle{deleg}(h)$ hands $\cA$ the secret behind $h$ and records $P_h$ in~$Q$; and $\oracle{use}(h,o)$, for $\adm(P_h,o)=1$, returns the transcript of an honest use and records~$o$ in~$\mathcal{T}$. $\cA$ wins if it outputs $(o^*,y^*)$ with $\Succ(o^*,y^*)=1$, $o^*\notin\bigcup_{P\in Q}I(P)$, and $o^*\notin\mathcal{T}$.
\end{definition}

Secrets leave honest parties only through $\oracle{deleg}$, so the coalition holds exactly the delegated secrets. The win condition excludes the operations of honest uses, not only their transcripts. A game that excluded only transcripts would count as a win a new transcript of an honest operation, for instance one carrying a re-randomized signature. Ruling that out needs strongly unforgeable signatures, and such a transcript is not an off-purpose achievement. In the payment task, $\oracle{deleg}$ hands a narrowed token to a possibly malicious assistant, and $\oracle{use}$ is an honest payment that the adversary observes.

\begin{definition}[Game \plgame{CONF-IND}]\label{def:conf-ind}
In Game \plgame{CONF-IND} (Figure~\ref{fig:games-conf}, right), the setup and key queries are as in \plgame{CONF-UP}. $\cA$ then names a challenge input $x$ (a plaintext, a PRF point) and receives either the objects the scheme produces for~$x$ (real) or the output of a simulator $\Sim$ (ideal). $\Sim$ is given the queried purposes, their real keys, the length $|x|$, and the \emph{entitled leakage} $\Leak_Q(x)$, the results on $x$ of the operations the queried purposes intend, but not $x$ itself. For functional encryption $\Sim$ gets the values $f_i(x)$ and outputs a ciphertext; for a PRF, the value at a non-admissible point is replaced by a uniform one. The scheme is \emph{indistinguishability-confining} for $I$, with a PPT simulator $\Sim$ that is part of the claim, if every PPT $\cA$ guesses the hidden bit with probability at most $\frac12+\negl(\lambda)$.
\end{definition}

\paragraph{Known primitives as instances.}
Three primitives fit the syntax with an exact proxy ($A=I$), and for each the game becomes the primitive's own security notion in a key-query form.
\begin{itemize}
  \item \emph{Functional signatures}~\cite{boyle2014functional}: a key $\sk_f$ derived from a master key signs exactly the messages in the range of~$f$. Here $\cP$ is functions~$f$, $A(f)=\range(f)$, and $\Succ(m,\sigma)=\Verify(\vk,m,\sigma)$. Game \plgame{CONF-UP} is unforgeability outside $\bigcup_i\range(f_i)$ against an adversary with key queries and no signing oracle (Proposition~\ref{prop:fs}).
  \item \emph{Constrained PRFs}~\cite{boneh2013constrained}: a key $k_C$ evaluates a pseudorandom function $F_k$ only at points with $C(x)=1$. Here $\cP$ is constraints~$C$ and $A(C)=\{x : C(x)=1\}$. Game \plgame{CONF-IND} with a uniform ideal value is pseudorandomness at non-admissible points against constrained-key queries made before the challenge. The standard constrained-PRF game also has an evaluation oracle and key queries after the challenge, so the proposition concerns this weaker notion.
  \item \emph{Functional encryption} (FE)~\cite{boneh2011fe}: a key $\sk_f$ recovers $f(x)$ from an encryption of~$x$ and nothing more. Here $\cP$ is functions~$f$, the admissible operation under $f$ is ``learn $f(x)$'' with entitled leakage $f(x)$, and $\Make$ is encryption. Game \plgame{CONF-IND} with the simulator given $\{f_i(x)\}$ is non-adaptive FE simulation security, with real setup and keys and only the challenge simulated. This is a mapping of definitions, not a proved equivalence, and it inherits FE's boundary: simulation-secure FE for general functions is impossible against unbounded collusion, even with non-adaptive queries~\cite{agrawal2013fe}, so the instance is meaningful for bounded collusion.
\end{itemize}

\begin{proposition}[Functional signatures]\label{prop:fs}
Fix a functional signature scheme $\mathsf{FS}$ and its instance $\mathsf{PLC}(\mathsf{FS})$ above. The instance is unpredictability-confining for $I$ (Game~\plgame{CONF-UP}) iff $\mathsf{FS}$ is unforgeable against key-query adversaries.
\end{proposition}

\begin{proposition}[Constrained PRFs]\label{prop:cprf}
For a constrained PRF~$F$, the instance $\mathsf{PLC}(F)$ above is indistinguishability-confining for $I$ (Game~\plgame{CONF-IND}, one challenge point, with the simulator that outputs the entitled leakage if it is nonempty and a uniform value otherwise) iff $F$ is pseudorandom at non-admissible points against constrained-key queries made before the challenge, with no evaluation oracle.
\end{proposition}

Both proofs translate one game into the other without loss (Appendix~\ref{app:propositions}).

\subsection{Policy-Based Signatures as an Instance}\label{sec:signing}

The build key of Section~\ref{sec:intro} is an instance of Proposition~\ref{prop:fs}, not a new construction, and it shows what stating the game against~$I$ adds. The capability is signing: $\Op$ is messages and $\Succ(m,\sigma)=\Verify(\vk,m,\sigma)$. Purposes are predicates $P$ over messages (``a well-formed project-$X$ artifact''), with the exact proxy $A(P)=I(P)=\{m : P(m)=1\}$. A policy-based signature~\cite{bellare2014policy} is then a PLC: $\Constrain(\msk,P)$ returns a key that signs exactly the messages satisfying~$P$, and $\Use$ signs. Policy-based signatures have no delegation algorithm, and we construct no $\Attenuate$.

\begin{corollary}[Signing-instance confinement]\label{cor:signing}
If the underlying policy-based signature is unforgeable against key-query adversaries, this PLC is unpredictability-confining for $I$ (Game~\plgame{CONF-UP}).
\end{corollary}

\begin{proof}
As for Proposition~\ref{prop:fs} with $\{m : P(m)=1\}$ in place of $\range(f)$: a \plgame{CONF-UP} winner outputs $(m^*,\sigma^*)$ with $\Verify(\vk,m^*,\sigma^*)=1$ and $P_i(m^*)=0$ for every queried~$P_i$, a key-query forgery.
\qed
\end{proof}

A functional signature realizes the same instance when $\{m : P(m)=1\}$ has an efficient encoding as the range of a function. The exact proxy $A(P)=I(P)$ is an assumption about the predicate, and it can fail. Suppose the policy language can express only a coarser predicate $\tilde P$ that checks that an artifact's name starts with the project prefix. A malicious build with that name is in $A(\tilde P)$ but not in $I(P)$, the false-permit region of Figure~\ref{fig:fidelity}. The policy-based signature is still unforgeable for~$\tilde P$, but the PLC is not confining for~$I$. Game~\plgame{CONF-UP}, stated against $I$, records the failure; the primitive's own game, stated against the predicate, does not.

\subsection{Mediation and Accountability}\label{sec:mediation-accountability}

\paragraph{Mediation.}
Under mediation a trusted component, not hardness, keeps use in purpose. The component may be a hardware security module (HSM) that stores keys, a trusted execution environment (TEE) that runs code in an attested enclave, or a verifier that checks results. The game models it as an oracle, and that interface is the integrity assumption an instantiation must justify.

\begin{definition}[Game \plgame{MED}]\label{def:med}
In Game \plgame{MED} (Figure~\ref{fig:games-deleg-med}, right), $\cA$ controls the host and creates monitors $M_i$ through $\algo{new}(P_i)$ for purposes of its choice, in one of two modes.
\begin{itemize}
  \item \emph{Monitor-held}: the secret $\sk_{P_i}$ never leaves~$M_i$, which runs $\Use$ only on operations admissible for~$P_i$. $\cA$ wins if it outputs $(o^*,y^*)$ with $\Succ(o^*,y^*)=1$, $o^*\notin\bigcup_i I(P_i)$, and $y^*$ not returned by any monitor.
  \item \emph{Verifier-gated}: the secret is a bearer object (usable by whoever holds it) given to~$\cA$, and the monitor is an honest verifier running a public acceptance algorithm $\Acc(\pp,o,y)$ that the scheme specifies. $\Acc$ reads the purpose from~$y$ (a certificate, a chain), checks that it admits~$o$, accepts honest uses of admissible operations, and accepts only if $\Succ(o,y)=1$. $\cA$ wins if the verifier accepts an off-purpose operation, $\Acc(\pp,o^*,y^*)=1$ with $o^*\notin\bigcup_i I(P_i)$.
\end{itemize}
The scheme is \emph{mediation-secure} if $\Adv{med}(\cA):=\Pr[\cA\text{ wins}]$ is negligible for all PPT~$\cA$.
\end{definition}

Monitor-held mediation turns ``the monitor honors a flag'' into ``the secret is unusable outside an attested boundary'', and an instantiation must justify that the secret stays inside (attestation, sealing). Verifier-gated mediation, which the token uses at redemption, needs only an honest verifier. The deployed purpose flags of Section~\ref{sec:intro} are mediation: a PKCS\#11 token is a monitor that holds the key, and an X.509 relying party that checks the key-usage field is verifier-gated, as is a server that checks OAuth access tokens or JSON Web Tokens (JWTs)~\cite{rfc5280,pkcs11v3,rfc6749,rfc7519,rfc8707,rfc9449}. Neither flag confines the key: once the check is bypassed, the key achieves off-purpose operations and wins \plgame{CONF-UP}~\cite{bortolozzo2010pkcs11}.

\begin{remark}[Verifier-gated mediation and confinement]\label{rem:med-conf}
In the verifier-gated mode $\cA$ holds the secrets of the purposes passed to $\algo{new}$, and an accepted result is valid. A gated win is therefore a \plgame{CONF-UP} win with $Q$ those purposes, and unpredictability confinement for $I$ implies verifier-gated mediation security for~$I$. The converse fails: an X.509 key whose certificate the adversary cannot forge is gated-secure, since the relying party reads the key usage from the signed certificate, yet the key signs off-purpose messages and is not confining. The gated mode asks only that a forbidden result not be accepted, and assumes that the verifier runs $\Acc$ before honoring a result.
\end{remark}

\paragraph{Accountability.}
Under accountability, use beyond a stated bound is possible but attributable: issuance binds an identity, and $\Trace$ resolves a log of accepted uses to it. The regime governs how often a secret is used, not which operations. Whether each use lies in $A(P)$ is left to confinement and mediation, so here each accepted use is assumed in-purpose; an extra use is permitted and attributed, not prevented.

\begin{definition}[Traceable limitation and non-frameability]\label{def:trace}
Extend the intended use with a \emph{use bound}: purpose $P$ intends at most $k_P$ uses of a secret. The game is the delegation-extended game for the accountable variant, with these changes.
\begin{itemize}
  \item Honest accounts are created with $\AKg$, and $\oracle{grant}$ binds the issued secret to an account's identity~$\id$.
  \item $L$ is the issuer's log of accepted uses and $R$ its record store when $\cA$ halts. Each use is logged once, in a form the scheme fixes, so a resubmitted use is a replay and does not count. A corrupt verifier lets $\cA$ submit uses, which enter $L$ only if they pass the scheme's public checks.
  \item For non-frameability the issuer may be corrupt; then $\cA$ writes $R$, but honest accounts keep their account keys.
\end{itemize}
A secret \emph{over-spends} if $L$ contains more than $k_P$ uses of it.
\begin{itemize}
  \item \emph{Traceable limitation}: if a secret $\iid$ bound to $\id$ over-spends, whoever spent it, then $\Trace(\pp,\iid,L,R)=\id$ except with negligible probability.
  \item \emph{Non-frameability}: except with negligible probability, there is no $\iid$ for which $\Trace(\pp,\iid,L,R)$ names an honest identity to which $\iid$ was never issued, or one to which it was issued that neither shared nor over-spent it.
\end{itemize}
\end{definition}

Non-frameability is per secret: an identity that delegates one secret is still protected on the others. Attribution follows possession, so a delegated secret that over-spends traces to the identity it was issued to, however often it was delegated. For a one-use token ($k_P=1$), traceable limitation is ``a double-spend identifies the withdrawing account''~\cite{chaum1988ecash,brands1993cash}, and non-frameability is e-cash exculpability. Traitor tracing is the many-user analogue~\cite{chor1994tracing}.

\paragraph{Relations among the regimes.}
Table~\ref{tab:regimes} in Section~\ref{sec:discussion} summarizes the regimes, and Remark~\ref{rem:med-conf} relates confinement to gated mediation. A purpose has a \emph{content} part (which operations) and may have a \emph{count} part (how many times). Confinement and mediation can enforce the content part. A stateless check cannot enforce the count part, because a verifier without state cannot tell a second presentation of a bearer secret from the first; the count part needs a stateful monitor or accountability over a log.

\section{A Purpose-Limited Payment Token}\label{sec:token}

The token realizes the invoice task and uses all three regimes on one secret. It is an off-line e-cash coin (the bank is not contacted at payment time) whose bank signature also covers a maximal purpose $P_0$ and a root public key. A holder narrows the purpose by appending certificates to a chain that starts at that key. At redemption the merchant checks the chain, checks that the payment is admissible for the chain's last purpose, and checks an e-cash opening that identifies the withdrawing account if the coin is spent twice. The building blocks are known: the representation coin of Brands~\cite{brands1993cash}, e-cash double-spend equations~\cite{chaum1988ecash}, and public-key delegation in the tradition of SPKI/SDSI (the Simple Public Key Infrastructure and Simple Distributed Security Infrastructure)~\cite{rfc2693}.

\subsection{Capability and Purpose Space}\label{sec:token-purposes}

Operations $\redeem(a,M',t)$ and purposes $P=(M,C,T)$ are those of the example in Section~\ref{sec:purposes}, with use bound $k_P=1$, which accountability enforces over the deposit log. Merchant identifiers form a tree: a merchant, its outlets as children, and a wildcard~$*$ at the root. Write $M'\preceq M$ if $M'$ is $M$ or a descendant of~$M$. Then
\[
  \adm\bigl((M,C,T),\redeem(a,M',t)\bigr)=[\,M'\preceq M\wedge a\le C\wedge t\in T\,].
\]
Caps range over $[0,\infty]$, and windows $T=[t_0,t_1]$ are ordered by inclusion. Refinement is componentwise, $(M',C',T')\sqsubseteq(M,C,T)$ iff $M'\preceq M$, $C'\le C$, and $T'\subseteq T$, and makes purposes a bounded lattice whose meet and join are componentwise: minimum and maximum for caps, intersection and interval hull for windows, and for merchants the deeper of two comparable nodes, or an empty merchant if they are incomparable, and the nearest common ancestor. $\top=(*,\infty,(-\infty,\infty))$ admits every operation, and all purposes with an empty merchant or an empty window are identified with~$\bot$, so $A(\bot)=\emptyset$; the join with~$\bot$ is $\bot\sqcup P=P$ rather than componentwise. Each conjunct of $\adm$ is downward closed in its order, so $P'\sqsubseteq P$ implies $A(P')\subseteq A(P)$, and a verifier decides $P_i\sqsubseteq P_{i-1}$ by three efficient tests. The proxy is the purpose predicate itself, so $A=I$ and the confinement claims need no soundness hypothesis.

\subsection{Building Blocks}\label{sec:token-blocks}

\paragraph{Group and hash.}
Let $\G$ be a group of prime order $q$ with generators $g_1,g_2$ such that $w_2:=\log_{g_1}g_2$ is unknown to all parties, for instance with $g_2$ derived from a hash of a public string; the hardness assumption is discrete log (DL) in~$\G$. A \emph{representation} of $X\in\G$ is a pair $(r_1,r_2)$ with $X=g_1^{r_1}g_2^{r_2}$, and two representations of one element reveal $w_2$, the basis of Brands' e-cash~\cite{brands1993cash}. A hash function $H:\{0,1\}^*\to\Zq$ is modeled as a random oracle~\cite{bellare1993random} in Theorems~\ref{thm:trace} and~\ref{thm:nonframe} and assumed only collision-resistant in Theorem~\ref{thm:atten}.

\paragraph{Signatures and proofs.}
The bank signs with an EUF-CMA (existentially unforgeable under chosen-message attack) signature $(\Sign,\Verify)$~\cite{goldwasser1988digital} and key pair $(\bsk,\bvk)$. Delegation certificates and account records use an EUF-CMA signature $\Sdel=(\Kg,\Sg,\Vf)$; Schnorr signatures over~$\G$ are a natural choice, but their unforgeability is proved in the random-oracle model~\cite{pointcheval2000security}. At issuance the holder proves knowledge of the coin identity $u$ and the coin secret~$s$, written $\pi\gets\mathrm{PoK}\{(u,s):Y=g_1^u\wedge A=(Yg_2)^s\}$: a Schnorr proof~\cite{schnorr1991efficient} for the conjunction, a three-move $\Sigma$-protocol with one challenge, made non-interactive by the Fiat--Shamir transform~\cite{fiat1986prove} over~$H$, whose input contains the full statement. In the random-oracle model the proof is zero-knowledge (a simulator programs~$H$) and extractable (rewinding to the query that fixed the challenge yields the witness)~\cite{pointcheval2000security}. Every signed message and hash input is canonically encoded and begins with a tag for its role (bank signature, certificate, spend authorization, withdrawal record, proof challenge, redemption challenge), so no signature or hash value serves two roles.

\paragraph{Design choices.}
The coin uses two generators and a per-coin pad commitment~$D$, the first message of a $\Sigma$-protocol, as in Brands' scheme and compact e-cash~\cite{brands1993cash,camenisch2005compact}. A one-generator coin with a Schnorr nonce gives a trace equation but no simulatable one-spend transcript, since its nonce commitment is certified before the challenge is known; the pad commitment lets the non-frameability reduction choose the opening first and solve for~$D$. Attenuation uses a \emph{public-key delegation chain}~\cite{rfc2693} rather than macaroons, which chain message authentication codes (MACs) under a root key that every verifier, here every merchant, would hold and could use to mint looser caveats~\cite{birgisson2014macaroons}.

\subsection{The Scheme}\label{sec:token-scheme}

\begin{figure}[t]
  \centering
  \small
\begin{tabular}{@{}p{4.7cm}@{\hspace{1mm}}c@{\hspace{1mm}}p{3.9cm}@{}}
  \textbf{Holder} $U$ (holds $\ask_U$) & & \textbf{Bank} (holds $\bsk$, store $R$)\\
  \midrule
  $u,s,a_1,a_2\gets\Zq$;\ $(\sk_0,\vk_0)\gets\Kg$ & & \\
  $Y:=g_1^u$;\ $A:=(Yg_2)^s$ & & \\
  $D:=g_1^{a_1}g_2^{a_2}$ & & \\
  $\pi\gets\mathrm{PoK}\{(u,s):$ & & \\
  \quad$Y=g_1^u\wedge A=(Yg_2)^s\}$ & & \\
  $\rho_w:=\Sg_{\ask_U}(Y,A,D)$ & & \\
  & \sendright{$\avk_U,\ P_0,\ Y,\ A,\ D$\\$\vk_0,\ \pi,\ \rho_w$} & \\
  & & check $\pi$, and $\rho_w$ under $\avk_U$\\
  & & refuse if $A=1$ or $(A,D)$ in $R$\\
  & & store $(\avk_U,Y,A,D,\rho_w)$ in $R$\\
  & & $\sigma:=\Sign_{\bsk}(A,D,P_0,\vk_0)$\\
  & \sendleft{$\sigma$} & \\
  \multicolumn{3}{@{}l}{$\sk_P:=(u,s,a_1,a_2,\sk_0,\vk_0,A,D,\sigma,P_0)$}\\
\end{tabular}

  \Description{Protocol diagram of issuance between the holder U and the bank. The holder samples per-coin secrets, forms Y, A, D, a proof of knowledge pi and the withdrawal record rho_w, and sends them with avk_U, P_0 and vk_0 to the bank; the bank checks the proof and the record, refuses A=1 or a repeated (A,D), stores the record, and returns the signature sigma.}
  \caption{Issuance ($\Constrain$) for maximal purpose $P_0$. The account-signed record $\rho_w$ is what tracing resolves to an account.}
  \label{fig:token-protocols}
\end{figure}

\begin{figure}[t]
  \centering
  \small
\begin{tabular}{@{}p{4.7cm}@{\hspace{1mm}}c@{\hspace{1mm}}p{3.9cm}@{}}
  \textbf{Spender} (holds $\ch$, $\sk_n$, $u,s,a_1,a_2$) & & \textbf{Merchant} $M'$ (holds $\bvk$)\\
  \midrule
  & & pick a nonce $N$\\
  & \sendleft{$N$} & \\
  $N_U\gets\{0,1\}^\lambda$ & & \\
  $d:=H(A,D,\vk_n,M',a,t,N,N_U)$ & & \\
  $r_1:=usd+a_1$;\ $r_2:=sd+a_2$ & & \\
  $\rho:=\Sg_{\sk_n}(d)$ & & \\
  & \sendright{$A,\ D,\ \ch,\ N_U,\ a,\ t$\\$r_1,\ r_2,\ \rho$} & \\
  & & recompute $d$\\
  & & accept iff (V1)--(V5) hold\\
  & & bank re-checks,\newline logs $(A,D,d,r_1,r_2)$\\
\end{tabular}

  \Description{Protocol diagram of redemption between the spender and the merchant. The merchant sends a nonce N; the spender picks N_U, computes the challenge d, the opening r_1, r_2 and the spend authorization rho, and sends them with A, D, the chain, N_U, a and t; the merchant recomputes d, checks (V1)-(V5), and the bank re-checks and logs the redemption.}
  \caption{Redemption ($\Use$) of $\redeem(a,M',t)$ with a chain ending at $(P_n,\vk_n)$; $d$ binds the operation.}
  \label{fig:token-redemption}
\end{figure}

\paragraph{Issuance.}
Holder $U$ registers, under the bank's know-your-customer (KYC) checks, a verification key $\avk_U$ and keeps the signing key~$\ask_U$; there is no long-term discrete-log identity. To withdraw a coin for a maximal purpose $P_0$ (Figure~\ref{fig:token-protocols}), $U$ samples a fresh \emph{coin identity} $u$ with pseudonym $Y=g_1^u$, a coin secret~$s$, pads $a_1,a_2$, and a root delegation key pair $(\sk_0,\vk_0)$, all independent per coin. It forms the \emph{coin} $A=(Yg_2)^s=g_1^{us}g_2^s$ and the \emph{pad commitment} $D=g_1^{a_1}g_2^{a_2}$, proves knowledge of $(u,s)$, and signs the \emph{withdrawal record} $\rho_w$ on $(Y,A,D)$. The bank checks both, stores the record, and signs $\sigma$ on $(A,D,P_0,\vk_0)$; the holder keeps $\sk_P=(u,s,a_1,a_2,\sk_0,\vk_0,A,D,\sigma,P_0)$, and the empty chain authorizes spending at~$P_0$. One withdrawal creates one \emph{instance} $(A,D)$, and the secret $\sk_P$ is everything the holder keeps for it. Each check at issuance serves tracing. The proof's first conjunct keeps $Y$ free of a $g_2$-component, without which a double-spend could trace to a value matching no pseudonym; the second ties $A$ to~$u$. The bank refuses $A=1$, which forces $s\neq0$, as tracing needs: $s=0$ gives $A=1$, and $A=1$ with $s\neq0$ would need $u=-w_2$. It refuses an instance $(A,D)$ already in~$R$, so each instance has one record, but not a repeated~$A$ with a new~$D$, since tracing operates on the instance $(A,D)$ and each withdrawal is debited separately.

\paragraph{Attenuation.}
A holder of the key $\sk_{\mathit{cur}}$ at the end of a chain for purpose $P_{\mathit{cur}}$ narrows it to $P'\sqsubseteq P_{\mathit{cur}}$ by drawing $(\sk',\vk')\gets\Kg$ and signing the certificate $\delta:=\Sg_{\sk_{\mathit{cur}}}(\langle P',\vk'\rangle)$. The authorization for $P'$ is the chain
\[
  \ch=\bigl((P_0,\vk_0,\sigma),(P_1,\vk_1,\delta_1),\ldots,(P_n,\vk_n,\delta_n)\bigr),
  \qquad P_0\sqsupseteq P_1\sqsupseteq\cdots\sqsupseteq P_n=P',
\]
handed to the delegate with the terminal key $\sk_n$ and the coin secrets $(u,s,a_1,a_2)$. In the running example the account withdraws at $P_0=(M,50,T)$ and hands the coin to the assistant, which gives its sub-contractor one certificate for $P_1=(M,20,T)$ and the key~$\sk_1$.

\paragraph{Redemption.}
To redeem $o=\redeem(a,M',t)$ with a chain ending at $(P_n,\vk_n)$ (Figure~\ref{fig:token-redemption}), the spender answers the merchant's nonce $N$ with its own nonce $N_U$, sets the challenge $d:=H(A,D,\vk_n,M',a,t,N,N_U)$, and sends the opening $r_1:=usd+a_1$, $r_2:=sd+a_2 \pmod q$ and the spend authorization $\rho:=\Sg_{\sk_n}(d)$ together with $(A,D,\ch,N_U,a,t)$. The merchant recomputes $d$ and accepts iff
\begin{itemize}
  \item[(V1)] $\Verify_{\bvk}((A,D,P_0,\vk_0),\sigma)=1$;
  \item[(V2)] for $i=1,\ldots,n$: $\Vf_{\vk_{i-1}}(\langle P_i,\vk_i\rangle,\delta_i)=1$ and $P_i\sqsubseteq P_{i-1}$;
  \item[(V3)] $\adm(P_n,o)=1$;
  \item[(V4)] $\Vf_{\vk_n}(d,\rho)=1$;
  \item[(V5)] $g_1^{r_1}g_2^{r_2}=A^d\cdot D$.
\end{itemize}
On deposit the bank re-runs (V1)--(V5) on the transcript, so a corrupt merchant cannot log a tuple that fails them, adds $(A,D,d,r_1,r_2)$ to its deposit log~$L$, and discards a tuple already present as a replay. Checks (V1)--(V2) authenticate the chain, (V3) enforces its terminal purpose, (V4) shows that the spender holds the terminal key, so it cannot present a shorter chain ending at a wider purpose, and (V5) is the traceable opening, which an honest spender satisfies since $A^dD=g_1^{usd+a_1}g_2^{sd+a_2}$.

The merchant identity $M'$ enters $d$ and is compared with $P_n$ in (V3). That the party running the checks is the merchant $M'$ it claims to be, for example because the spender reaches it over a channel authenticated to~$M'$, is part of the trust placed in the verifier and so belongs to mediation. Theorem~\ref{thm:atten} bounds which operations can be accepted; the verifier's honesty and authenticity ensure that an accepted operation is the one that took place.

\paragraph{Identity binding and delegation.}
Tracing follows a chain of bindings from two openings to $u$ (the issuance proof ties $A$ to~$u$), to $Y$, and to the account (the record $\rho_w$ binds $Y$ and $(A,D)$ to it). The coin secrets belong to one coin, so attenuation hands over one single-use instance, narrowed in purpose, and nothing long-term: a delegate learns that coin's~$u$, since an opening needs both $us$ and~$s$, but nothing about the account's other coins or its signing key.

\paragraph{Tracing, $\Trace(\pp,(A,D),L,R)$.}
If $L$ contains two tuples for the instance $(A,D)$ with $r_2\neq r_2'$, $\Trace$ computes $u:=(r_1-r_1')(r_2-r_2')^{-1}\bmod q$ and $Y:=g_1^u$, and otherwise outputs~$\bot$; this is correct because $r_1-r_1'=us(d-d')$ and $r_2-r_2'=s(d-d')$. It outputs the account of the record $(\avk_U,Y,A,D,\rho_w)\in R$ if it is the only record for $(A,D)$, carries~$Y$, and verifies, $\Vf_{\avk_U}((Y,A,D),\rho_w)=1$, and $\bot$ otherwise. The signature check matters when the bank is corrupt and writes~$R$.

\subsection{Security}\label{sec:token-security}

\begin{figure}[t]
  \centering
  \begin{tikzpicture}[
    font=\small,
    node distance=4mm,
    phase/.style={draw, rounded corners=2pt, fill=gray!10, text width=6.7cm,
                  align=left, inner sep=4pt},
    regime/.style={draw, rounded corners=2pt, fill=gray!25, minimum width=2.5cm,
                   minimum height=5mm, font=\small\itshape},
    flow/.style={-{Stealth[length=2mm]}, thick},
  ]
  \node[phase] (issue) {\textbf{1. Issuance} (bank, holder).
    Holder forms coin $A=(Yg_2)^s$, pad~$D$, root key~$\vk_0$; bank signs
    $\sigma=\Sign(A,D,P_0,\vk_0)$ and stores the account-signed withdrawal record.};
  \node[phase, below=of issue] (att) {\textbf{2. Attenuation} (delegation chain).
    Certificate $\delta_i=\Sg_{\sk_{i-1}}\langle P_i,\vk_i\rangle$ narrows
    $P_0\sqsupseteq P_1\sqsupseteq\cdots\sqsupseteq P_n$; a bearer cannot widen.};
  \node[phase, below=of att] (red) {\textbf{3. Redemption} (spender, merchant).
    Send $(A,D,\ch,N_U,a,t,r_1,r_2,\rho)$; the merchant checks (V1)--(V5).};
  \node[phase, below=of red] (reuse) {\textbf{Reuse.}
    Two redemptions of one instance $(A,D)$ give $\Trace$:
    $u=(r_1-r_1')(r_2-r_2')^{-1}$, the coin identity, traced to the account.};

  \draw[flow] (issue) -- (att);
  \draw[flow] (att)   -- (red);
  \draw[flow] (red)   -- (reuse);

  \node[regime, right=7mm of att]   (conf) {Confinement};
  \node[regime, right=7mm of red]   (med)  {Mediation};
  \node[regime, right=7mm of reuse] (acc)  {Accountability};

  \draw[flow] (att)   -- (conf);
  \draw[flow] (red)   -- (med);
  \draw[flow] (reuse) -- (acc);
\end{tikzpicture}
  \Description{Flow diagram of four stacked boxes connected by downward arrows: 1. Issuance (bank, holder), 2. Attenuation (delegation chain), 3. Redemption (spender, merchant), and Reuse. Arrows lead from Attenuation to a box labelled Confinement, from Redemption to a box labelled Mediation, and from Reuse to a box labelled Accountability.}
  \caption{Phases of the token and the result that covers each. Attenuation is confinement (Theorem~\ref{thm:atten}), redemption is verifier-gated mediation (Corollary~\ref{cor:med}), and reuse is held to account (Theorems~\ref{thm:trace} and~\ref{thm:nonframe}).}
  \label{fig:token}
\end{figure}

\begin{table}[t]
  \caption{Security results for the token, each proved separately. CR: collision resistance; ROM: random-oracle model. ``Standard'' describes the reduction; with Schnorr certificates, whose unforgeability is proved in the ROM, Theorem~\ref{thm:atten} also rests on the ROM.}
  \label{tab:token-results}
  \centering\small\setlength\tabcolsep{4pt}
  \begin{tabular}{@{}llll@{}}
    \toprule
    Result & Regime & Assumptions & Model \\
    \midrule
    Thm.~\ref{thm:atten} attenuation & confinement & EUF-CMA ($\Sdel$, bank), CR of $H$ & standard \\
    Thm.~\ref{thm:trace} traceability & accountability & DL & ROM \\
    Thm.~\ref{thm:nonframe} non-frameability & accountability & DL, EUF-CMA (accounts) & ROM \\
    Cor.~\ref{cor:med} mediation & mediation, gated & as Thm.~\ref{thm:atten}; honest verifier & standard \\
    Thm.~\ref{thm:unlink} unlinkability & privacy (variant) & well-formed bank key & perfect \\
    \bottomrule
  \end{tabular}
\end{table}

Each regime has its own result, proved separately on the same token (Figure~\ref{fig:token}). Theorem~\ref{thm:atten} (confinement) says that narrowing cannot be undone: a coalition achieves nothing outside the purposes delegated to it, apart from operations that honest holders performed. Theorems~\ref{thm:trace} and~\ref{thm:nonframe} (accountability) say that a double-spend names the withdrawing account and that an honest account is not named for an instance it never withdrew, or withdrew but neither delegated nor spent twice. Corollary~\ref{cor:med} gives verifier-gated mediation, and Table~\ref{tab:token-results} lists the assumptions. There is no joint security statement (Section~\ref{sec:discussion}). For confinement, an off-purpose operation is an accepted redemption outside the purposes delegated to the adversary; for accountability, the target is reuse, with use bound~1.

\begin{lemma}[Representation binding]\label{lem:rep}
If a PPT algorithm, given $(g_1,g_2)$, outputs $X\in\G$ and two pairs $(r_1,r_2)\neq(\tilde r_1,\tilde r_2)$ with $g_1^{r_1}g_2^{r_2}=X=g_1^{\tilde r_1}g_2^{\tilde r_2}$, then $w_2=(r_1-\tilde r_1)(\tilde r_2-r_2)^{-1}\bmod q$, so it solves discrete log.
\end{lemma}

\begin{proof}
From $g_1^{r_1-\tilde r_1}=g_2^{\tilde r_2-r_2}$ we get $r_1-\tilde r_1\equiv w_2(\tilde r_2-r_2)\pmod q$. If $\tilde r_2=r_2$ then $r_1=\tilde r_1$, a contradiction, so $\tilde r_2-r_2$ is invertible.
\qed
\end{proof}

\begin{theorem}[Attenuation unforgeability]\label{thm:atten}
In the game of Definition~\ref{def:conf-up-deleg}, assume the delegation and bank signatures are EUF-CMA and $H$ is collision-resistant. Let the adversary's $\oracle{deleg}$ queries for a coin $A$ delegate the purposes $G=\{P_{n_1},\ldots,P_{n_k}\}$, each with its chain, terminal key, and the coin secrets. Then, except with negligible probability, the adversary obtains no accepted redemption of $A$ for an operation outside $\bigcup_j A(P_{n_j})$ other than the operations performed through $\oracle{use}$. Across coins the achievable set is the union over coins, since each redemption presents one coin. The reduction uses no random oracle, so the result holds in the standard model when both signatures are unforgeable in the standard model.
\end{theorem}

\begin{proof}
An accepted redemption presents a chain $\sigma,(P_1,\vk_1,\delta_1),\ldots,(P_m,\vk_m,\delta_m)$ with $P_i\sqsubseteq P_{i-1}$ for all~$i$ (V2), a spend authorization $\rho$ valid under $\vk_m$ (V4), and $o\in A(P_m)$ (V3). Purposes along the chain do not increase, so $A(P_m)\subseteq A(P_0)$, and no redemption exceeds the bank-signed $P_0$ unless $\sigma$ is forged.

Suppose the coalition wins, so $o\notin\bigcup_j A(P_{n_j})$, $o$ was not performed through $\oracle{use}$, and hence $P_m\not\sqsubseteq P_{n_j}$ for every~$j$ (otherwise $o\in A(P_m)\subseteq A(P_{n_j})$ by Definition~\ref{def:monotone}). Call a key \emph{honest} if an honest holder or delegator generated it. Secrets leave honest parties only through $\oracle{deleg}$, so the coalition holds the secret of an honest key only if that key was delegated, with purpose in~$G$. An honest-use transcript contains a spend authorization on its own challenge, which hashes its operation, so reusing it for $o$ requires a collision in~$H$. We show that the winning transcript contains a valid signature, a certificate $\delta_i$ or the spend authorization~$\rho$, under an honest key on a message its owner never signed.

\emph{Case 1: $\vk_m$ is honest and not delegated.} Then $\rho$ is a signature under $\vk_m$ on a challenge its owner never signed: the owner signed only the challenges of its own $\oracle{use}$ spends, whose operations differ from~$o$, so an equal challenge is a collision in~$H$.

Call the chain \emph{authentic} if the bank issued $\sigma$ and, for each~$i$, the honest owner of $\vk_{i-1}$ signed $\langle P_i,\vk_i\rangle$. An honest key is certified once, by its creator, for the purpose of its handle.

\emph{Case 2: $\vk_m$ is honest and delegated.} On an authentic chain $P_m$ is the purpose with which $\vk_m$ was delegated, so $P_m\in G$, a contradiction. So the chain is not authentic.

\emph{Case 3: $\vk_m$ is not honest.} The coalition may supply $\vk_m$ without knowing a secret key for it. An authentic chain ends at an honest key, so again the chain is not authentic.

In Cases 2 and 3, let $i^\star$ be the least index whose certificate was not honestly issued. Unless $\sigma$ is forged, $i^\star\ge1$ and the prefix up to $\vk_{i^\star-1}$ is authentic, so $\vk_{i^\star-1}$ is honest. If the coalition held $\sk_{i^\star-1}$, then $\vk_{i^\star-1}$ was delegated for $P_{i^\star-1}\in G$, and (V2) would give $P_m\sqsubseteq P_{i^\star-1}$, a contradiction. So the honest owner of $\vk_{i^\star-1}$ kept its key and never signed $\langle P_{i^\star},\vk_{i^\star}\rangle$, yet (V2) accepts $\delta_{i^\star}$.

A reduction guesses which honest key receives the forgery, or that it is the bank signature, plants its EUF-CMA challenge key there, answers every honest signature with its signing oracle, and outputs the forged $\delta_{i^\star}$, $\rho$, or~$\sigma$. With $q_k$ honest delegation keys, and $\Adv[X]{EUF}$ and $\Adv[H]{CR}$ the maximal advantages of a PPT adversary in forging~$X$ and finding a collision in~$H$,
\[
  \Adv{atten}(\cA)\le q_k\cdot\Adv[\Sdel]{EUF}+\Adv[\mathsf{bank}]{EUF}+\Adv[H]{CR}.
\]
The spend authorization (V4) rules out the \emph{truncation attack}, in which the coalition presents a shorter, all-honest chain ending at the key of a wider prefix, a key it was never given.
\qed
\end{proof}

\begin{theorem}[One-use traceability]\label{thm:trace}
This is traceable limitation for use bound~1, with $L$ the bank's deposit log, in which a use is a distinct tuple $(A,D,d,r_1,r_2)$, $R$ its withdrawal-record store, and the secret the instance $(A,D)$. Assume discrete log is hard in $\G$, in the random-oracle model. If $L$ contains two tuples for an instance $(A,D)$ issued under pseudonym $Y=g_1^u$, then, except with negligible probability, $\Trace(\pp,(A,D),L,R)$ recovers~$u$, hence~$Y$, and outputs the withdrawing account, whose record for $(A,D)$ is unique. No nonce freshness is assumed, and merchants may collude with the spender, though logged tuples pass (V1)--(V5). That a holder cannot forge a second certified instance from one withdrawal is unforgeability of the bank signature.
\end{theorem}

\begin{proof}
The reduction $\cB$ is given $(g_1,g_2)$, must output~$w_2$, and plays the bank. Dividing the two (V5) equations cancels~$D$:
\[
  g_1^{\Delta_1}g_2^{\Delta_2}=A^{d-d'},\qquad (\Delta_1,\Delta_2):=(r_1-r_1',\ r_2-r_2').
\]
Extracting $\pi$ by rewinding (accounted for below) gives $u$ and $s$, so $\cB$ knows the representation $\bigl(us(d-d'),\ s(d-d')\bigr)$ of $A^{d-d'}$. If $d=d'$, then $A^{d-d'}=1$ has the representation $(0,0)$, and since the tuples differ, $(\Delta_1,\Delta_2)\neq(0,0)$ is a second one, giving $w_2$ by Lemma~\ref{lem:rep}. If $d\neq d'$ and $(\Delta_1,\Delta_2)\neq\bigl(us(d-d'),s(d-d')\bigr)$, again $\cB$ outputs~$w_2$. Otherwise $\Delta_2=s(d-d')\neq0$, since the bank's check $A\neq1$ forces $s\neq0$ (by the check, not by sampling, since $s$ is adversarial), and $\Delta_1\Delta_2^{-1}=u$; so $\Trace$ outputs $\bot$ on $\Delta_2=0$ only for an off-form pair. Then $Y=g_1^u$ is the declared pseudonym, and the bank's record for $(A,D)$, unique and verifying because the bank checked both at issuance, names the account.

Call an opening pair \emph{off-form} if it is not the representation above; $\Trace$ fails only on an off-form pair. Let $\varepsilon$ be its probability and $q_H$ the number of random-oracle queries. As usual, the adversary queries $H$ on the input of every $\pi$ it submits, at the cost of one query each in $q_H$, so each accepted $\pi$ has a query that fixed its challenge; $\cB$ selects the query for this instance's $\pi$, whose input contains $(Y,A)$, rewinds the adversary to it, and reruns it with a fresh answer. By the general forking lemma~\cite{pointcheval2000security,bellare2006forking}, with probability at least $\varepsilon(\varepsilon/q_H-1/q)$ both runs produce an accepted $\pi$ with the same commitment and different challenges, giving $(u,s)$, together with an off-form pair, giving $w_2$. So $\varepsilon(\varepsilon/q_H-1/q)\le\Adv{DL}$, which gives
\[
  \Pr[\Trace\text{ fails}]\;\le\;\varepsilon\;\le\;\frac{q_H}{q}+\sqrt{q_H\cdot\Adv{DL}},
\]
negligible in~$\lambda$.
\qed
\end{proof}

\begin{theorem}[Non-frameability]\label{thm:nonframe}
This is the non-frameability of Definition~\ref{def:trace} for the token, per instance, with ``sharing'' meaning delegating that instance. Assume discrete log is hard in $\G$ and the account signature is EUF-CMA, in the random-oracle model. Then no PPT coalition, including a corrupt bank, merchants, and delegates of other instances, makes $\Trace$ name an honest account~$U$ for an instance that $U$ never withdrew, or withdrew but neither delegated nor spent twice, except with negligible probability.
\end{theorem}

\begin{proof}[sketch]
$\Trace$ names $U$ only through a record verifying under $\avk_U$. A record $U$ did not sign is a forgery of the account signature. For a record $U$ did sign, on an instance $U$ neither delegated nor double-spent, a reduction plants a discrete-log challenge as the coin's pseudonym $Y^*$ and simulates the one honest spend by choosing the opening first, solving for $D$, and programming the random oracle, so any framing of $U$ recovers $\log_{g_1}Y^*$; the corrupt-bank case is included. The full proof is in Appendix~\ref{app:nonframe}.
\qed
\end{proof}

\begin{corollary}[Verifier-gated mediation]\label{cor:med}
Under the assumptions of Theorem~\ref{thm:atten}, the token is mediation-secure in the verifier-gated mode, with the secret a bearer object and $\Acc$ the checks (V1)--(V5) run by the merchant or bank.
\end{corollary}

\begin{proof}
$\algo{new}(P)$ is $\oracle{grant}(P)$ followed by $\oracle{deleg}$, so a gated win is an accepted redemption of an operation outside $\bigcup_j I(P_{n_j})=\bigcup_j A(P_{n_j})$, the proxy being exact, with no $\oracle{use}$ queries. Theorem~\ref{thm:atten} bounds its probability; this is Remark~\ref{rem:med-conf} for the token.
\qed
\end{proof}

\subsection{Comparison with Related Designs}\label{sec:token-compare}

\begin{table}[t]
  \caption{Properties of related designs in their base form. Each property of the token appears in an earlier design; the token combines them. Public verif.: verifiers hold no key that lets them mint credentials. --: not provided.}
  \label{tab:compare}
  \centering\small
  \begin{tabular}{@{}lccccc@{}}
    \toprule
    Design & Attenuation & Public verif. & Use bound & Over-use traced & Payer anonymity \\
    \midrule
    Macaroons~\cite{birgisson2014macaroons} & $\checkmark$ & -- & -- & -- & -- \\
    SPKI/SDSI~\cite{rfc2693} & $\checkmark$ & $\checkmark$ & online check & -- & -- \\
    Brands e-cash~\cite{brands1993cash} & -- & $\checkmark$ & 1 & $\checkmark$ & $\checkmark$ \\
    Compact e-cash~\cite{camenisch2005compact} & -- & $\checkmark$ & wallet size & $\checkmark$ & $\checkmark$ \\
    Deleg.\ anon.\ cred.~\cite{belenkiy2009delegatable} & $\checkmark$ & $\checkmark$ & -- & -- & $\checkmark$ \\
    \midrule
    This token & $\checkmark$ & $\checkmark$ & 1 & $\checkmark$, per coin & -- \\
    Anonymous variant & -- & $\checkmark$ & 1 & $\checkmark$, per account & $\checkmark$ \\
    \bottomrule
  \end{tabular}
\end{table}

The token combines attenuation with public verification, a per-coin use bound with tracing, and non-frameability for coins the account never delegated. None of these properties is new, and the token does not provide anonymity together with delegation (Table~\ref{tab:compare}). Macaroons attenuate, but their verifiers can mint, which suits a service that both issues and verifies. SPKI/SDSI chains attenuate with public verification, and the token's chain is an instance of them; SPKI bounds use only through an online check with the issuer, which is stateful mediation, and does not trace over-use. Brands' e-cash and compact e-cash bound use and trace double-spenders while keeping honest payers anonymous, but have no delegation and no purpose beyond the coin's value. Delegatable anonymous credentials attenuate along anonymous chains but neither bound nor trace use. The anonymous variant below gives payer anonymity without delegation.

\subsection{An Anonymous Variant}\label{sec:token-anon}

The token is not anonymous: the bank sees $(A,D)$ at issuance, so it can link a redemption to the withdrawal. Hiding $(A,D)$ removes the per-coin record that tracing uses, and tracing must then go through a per-account value, which a delegate would learn (Section~\ref{sec:token-scheme}). An anonymous variant therefore drops delegation and replaces the per-coin identity by a per-account one, as in Brands' scheme~\cite{brands1993cash}: the account registers $Y=g_1^u$ once, which the bank refuses if $Yg_2=1$, since every coin would then be~$1$, and each coin is $A=(Yg_2)^s$ for a fresh $s\neq0$. The bank signs $(A,D,\vk_0)$ blindly. The maximal purpose stays authenticated through the key: the bank keeps one signing key for each of a fixed set of maximal purposes $P_0$, as e-cash keeps one key per denomination, and (V1) uses the key for~$P_0$.

\paragraph{Issuance in the variant.}
Issuance is Brands' withdrawal protocol~\cite{brands1993cash} with $\vk_0$ added to the signed message. For a third generator $g_0$, the key for $P_0$ is $x\in\Zq$ with public $h=g_0^x$, and each account holds $z=(Yg_2)^x$. A signature on $m=(A,D,\vk_0)$ is a tuple $\sigma=(z',\alpha',\beta',\gamma')$ with
\[
  g_0^{\gamma'}=h^{c'}\alpha',\qquad A^{\gamma'}=(z')^{c'}\beta',\qquad c'=H(m,z',\alpha',\beta'),
\]
the hash value read in $\Zq\setminus\{0\}$. In a session the bank sends $\alpha=g_0^{\eta}$ and $\beta=(Yg_2)^{\eta}$ for a fresh~$\eta$. The holder draws $s\neq0$, the pads, $(\sk_0,\vk_0)\gets\Kg$, and blinding values $e_1\in\Zq\setminus\{0\}$ and $e_2\in\Zq$; it forms $A$ and $D$, sets $z'=z^s$, $\alpha'=\alpha^{e_1}g_0^{e_2}$, and $\beta'=\beta^{se_1}A^{e_2}$, and sends $c=c'/e_1$. The bank answers $\gamma=cx+\eta$; the holder accepts iff $g_0^{\gamma}=h^c\alpha$ and $(Yg_2)^{\gamma}=z^c\beta$, and sets $\gamma'=e_1\gamma+e_2$. Redemption is as in Section~\ref{sec:token-scheme} with the empty chain, so $\rho=\Sg_{\sk_0}(d)$, and (V1) checks $\sigma$ under the key for~$P_0$. The key is \emph{well formed} if every account's $z$ is $(Yg_2)^x$ for the~$x$ with $h=g_0^x$; the bank can make this checkable by publishing $g_1^x$ and $g_2^x$ with a proof that they share the exponent of~$h$, so that each holder computes~$z$ itself. A bank that gave two accounts values $z$ with different exponents could tell their coins apart by~$z'$.

\begin{lemma}[Issuance view]\label{lem:issue-view}
Let $W=(s,a_1,a_2,\sk_0)$ be a holder's coin values, $m=(A,D,\vk_0)$ the message they determine for account~$Y$, and $S(m)$ the set of valid signatures on $m$ whose first component is $A^x$. For every bank strategy under a well-formed key and every view $v$ of a session in which the holder accepts,
\[
  \Pr[W=w,\ \sigma=\zeta\mid V=v]=\Pr[W=w]\cdot\frac{[\,\zeta\in S(m)\,]}{q}.
\]
So, given the bank's view, the coin values keep their honest distribution, and $\sigma$ is uniform on $S(m)$ and depends on the coin values only through~$m$.
\end{lemma}

\begin{theorem}[Payer unlinkability, anonymous variant]\label{thm:unlink}
In the anonymous variant under a well-formed key, let an adversary of unbounded power play the bank and all merchants, run two sequential issuance sessions with honest accounts of its choice under the key for one~$P_0$, and then see each coin redeemed once, in random order, for one public transaction $(M',a,t)$. It guesses the order with advantage~$0$. So issuance is unlinkable to redemption, and distinct single spends are mutually unlinkable, among coins issued for the same $P_0$ and redeemed for the same transaction.
\end{theorem}

\begin{proof}[sketch]
For a fixed accepting view, each pair of a coin value and a valid signature is reached by exactly one choice of $(e_1,e_2)$, which gives Lemma~\ref{lem:issue-view}; this is Brands' count of the holder's choices that match a withdrawal view to a payment view~\cite[Prop.~11]{brands1993cash}, and $\vk_0$ and the pads enter the holder's only message through~$c'$. Given the bank's view, for an account value $u$, which has $u+w_2\neq0$ by the registration check, a uniform $s\neq0$ makes $A=g_1^{s(u+w_2)}$ uniform over $\G\setminus\{1\}$; given $(A,D,d)$ and~$\sigma$, the opening $(r_1,r_2)$ is uniform over the $q$ representations of $A^dD$; and $\rho$ depends only on $\sk_0$ and~$d$. None of these depends on~$u$, so a redemption transcript is independent of the account, though not of~$s$, which $A$ determines. The full proofs are in Appendix~\ref{app:unlink}.
\qed
\end{proof}

The variant changes the accountability model. The bank never sees $(A,D)$, so it cannot store the per-coin record $\rho_w$; a double-spend recovers the registered value~$u$, which names the account~\cite{brands1993cash,camenisch2005compact}. That tracing rests on Brands' assumption that the protocol is \emph{restrictive}, so that every coin a holder obtains has $A=(Yg_2)^s$ for an $s\neq0$ it knows, a nonstandard assumption~\cite{brands1993cash,camenisch2005compact}, and the bank's unforgeability needs sequential issuance, since concurrent blind Schnorr issuance falls to a polynomial-time attack on the underlying ROS problem~\cite{schnorr2001blind,benhamouda2021ros}. The variant thus trades the per-coin identity and its non-frameable delegation for a per-account value, and the claim that one secret spans all three regimes is for the non-anonymous token. Combining anonymity with per-coin non-frameable delegation is open.

\section{Discussion}\label{sec:discussion}

\begin{table}[t]
  \caption{The enforcement regimes. All are stated against the same intended use~$I$; they differ in what happens to an off-purpose operation and in what the guarantee rests on.}
  \label{tab:regimes}
  \centering\small
\setlength\tabcolsep{4pt}
\begin{tabular}{@{}>{\raggedright}p{2.6cm}>{\raggedright}p{2.9cm}>{\raggedright}p{2.3cm}>{\raggedright\arraybackslash}p{3.4cm}@{}}
  \toprule
  Regime & a violation is & rests on & examples \\
  \midrule
  Confinement (UP)  & infeasible to produce & hardness & functional and policy-based sig., token attenuation \\
  Confinement (IND) & infeasible to learn & hardness & constrained PRF, FE \\
  Mediation, held   & refused by a monitor that holds the key & integrity of the monitor & HSM or TEE key, PKCS\#11 \\
  Mediation, gated  & rejected by a verifier & honest verifier & X.509 key usage, OAuth/JWT, token redemption \\
  Accountability    & possible; over-use is traced & log and $\Trace$ soundness & e-cash double-spend, traitor tracing \\
  \bottomrule
\end{tabular}

\end{table}

The confinement games coincide with existing notions by construction, the use bound of Definition~\ref{def:trace} is e-cash over-spend tracing, and the token is assembled from known parts. What the common specification changes is the target: each guarantee is stated against the intended use~$I$, not against the predicate a primitive enforces. The game of a functional or policy-based signature is stated against the range of $f$ or the policy, so it cannot express that the predicate itself permits too much; the coarse build-key predicate of Section~\ref{sec:signing} is unforgeable in the primitive's game and insecure in Game~\plgame{CONF-UP}. Stating every regime against one~$I$ also relates the regimes: verifier-gated mediation follows from unpredictability confinement and an honest verifier, and a use bound needs a stateful monitor or accountability. Table~\ref{tab:regimes} places mechanisms from functional cryptography, hardware key stores, e-cash, and capability tokens in this common form.

The token's results are proved one regime at a time. We prove neither a joint security statement for the token nor a general composition theorem, which would state when a scheme that meets one regime's game for the content of a purpose and another's for its count meets a combined game for the whole purpose. Such a theorem would give the specification a role beyond organizing existing notions; we leave it open.

\subsection{Limitations}\label{sec:limits}

\paragraph{The release problem.}
Cryptography governs access to a capability, not the use of released plaintext: once the analytics firm or model trainer holds cleartext, no key can take it back. This is the cryptographic analogue of the ``analog hole'' in digital rights management (DRM)~\cite{erickson2003drm}, and a formal impossibility would need a model of release. Functional encryption avoids release, since the plaintext never appears, but does not govern it. Release can be addressed only by mediation, with a trusted execution environment keeping plaintext inside an attested boundary, or by accountability after the fact. The ``train one model'' example is of this kind, and we construct no scheme for it.

\paragraph{Authorship of the proxy.}
A mechanism enforces the proxy~$A$, so someone must write $A$ and audit its fidelity to~$I$ (Section~\ref{sec:intro}). Purpose-limitation cryptography leaves this policy problem in place and moves the trust from ``the monitor honors the flag'' to ``the proxy is faithful to the purpose''. Given a faithful proxy, a confinement guarantee holds under the stated assumptions without a trusted monitor; given an unfaithful one, the mechanism enforces the wrong predicate.

\paragraph{Legal alignment.}
A scheme can enforce the specification part of GDPR Art.~5(1)(b), in that the secret cannot be used outside the declared proxy, but the compatible-use test of Art.~6(4) is a controller's context-dependent judgment~\cite{wp29-2013purpose}. M\"uhlhoff and Ruschemeier argue that the doctrine cannot be traced at machine-learning scale~\cite{muhlhoff2025purpose}. The framework does not resolve that question; it moves the difficulty to writing and auditing the proxy. This, and reviving purpose limitation in data-driven systems more broadly~\cite{finck2021reviving,wolf2021messaging}, is interdisciplinary future work.

\paragraph{Evolving purposes and further constructions.}
Purposes change over time, and usage control treats control as continuous rather than one-shot~\cite{park2004ucon}; how attenuation, revocation, and time-bounded purposes compose is open. So are a mediation construction that binds a decryption key to an attested trusted execution environment and an accountability construction that traces off-purpose use of a trained model.

\section{Related Work}\label{sec:related}

\paragraph{Unifying frameworks and delegatable credentials.}
Cryptographic agents~\cite{agrawal2015agents} give one schema-based security notion for functional encryption, obfuscation, and related secrecy primitives; they cover the indistinguishability side, without an unforgeable-capability form, a delegation lattice, or an accountability regime. Access control encryption~\cite{damgard2016ace} governs which parties may send and receive information, with a sanitizer in our mediation position, and so combines confinement and mediation for one capability. K\"usters, Truderung, and Vogt~\cite{kusters2010accountability} define accountability in general; traceable limitation is a purpose-specific instance over usage logs. Delegatable anonymous credentials~\cite{chase2006sok,belenkiy2009delegatable} give monotone attenuation along anonymous chains; Section~\ref{sec:token-compare} compares them, macaroons, SPKI/SDSI~\cite{rfc2693}, and e-cash with the token.

\paragraph{Confinement.}
Functional encryption and its relatives confine what a ciphertext reveals~\cite{sahai2005fuzzy,goyal2006abe,bethencourt2007cpabe,katz2008predicate,boneh2011fe,abdalla2015simple}; constrained and puncturable PRFs confine which inputs a key evaluates~\cite{boneh2013constrained,boyle2014functional,kiayias2013delegatable,sahai2014io}; functional, policy-based, attribute-based, and sanitizable signatures confine which messages a key signs~\cite{boyle2014functional,bellare2014policy,maji2011abs,ateniese2005sanitizable}; and conditional proxy re-encryption confines a delegated transformation~\cite{blaze1998proxy,ateniese2005proxy,weng2009conditional}. Each is an instance of a confinement game with the primitive's predicate as the proxy.

\paragraph{Mediation.}
X.509 and PKCS\#11 usage flags are the deployed ``key for a purpose'', mediation without confinement at the key level (Section~\ref{sec:mediation-accountability}); the attacks on PKCS\#11 show that the flags carry no cryptographic hardness~\cite{rfc5280,pkcs11v3,clulow2003pkcs11,bortolozzo2010pkcs11,cachin2009token}. Macaroons and the OAuth/JWT token family attenuate and bind bearer credentials checked by verifiers~\cite{birgisson2014macaroons,rfc6749,rfc8707,rfc7519,rfc9449}. Macaroons sit between hardness and mediation: attenuation is unforgeable to bearers, but the verifier enforces the meaning of a caveat and holds the root key.

\paragraph{Accountability.}
Traitor tracing attributes leaked keys~\cite{chor1994tracing,boneh2006traitor}, non-transferable credentials deter sharing~\cite{camenisch2001credentials}, and e-cash makes a coin a one-use token whose double-spend identifies the withdrawer~\cite{chaum1982blind,chaum1988ecash,brands1993cash,camenisch2005compact,camenisch2004bilinear}, the pattern the token adopts.

\paragraph{Purpose in access control, usage control, and law.}
Purpose-based access control~\cite{byun2005purpose,byun2008purpose} attaches intended purposes to data and access purposes to requests, and grants access when they comply; usage control (UCON)~\cite{park2004ucon} adds obligations, conditions, and continuous enforcement; Hippocratic databases~\cite{agrawal2002hippocratic}, policy languages~\cite{p3p2002,epal2003}, and sticky policies~\cite{karjoth2002epp,pearson2011sticky} carry purpose with the data. These lines specify an intended-use relation, often richer than ours since it can be stateful and depend on context, and enforce it through a reference monitor; in our terms they specify~$I$ by policy and enforce it by mediation. The difference we address is narrow: under confinement the secret itself enforces the proxy by hardness, so the limit holds when no monitor is present or the monitor is bypassed. Purpose has also been formalized for audit, as a plan~\cite{tschantz2012purpose} or a workflow~\cite{jafari2011purpose} against which behavior is tested after the fact, close in spirit to accountability; contextual integrity identifies the observable flow parameters (sender, recipient, information type, transmission conditions)~\cite{nissenbaum2004contextual,barth2006contextual}, the kind of observable over which an intended-use relation is written. The legal texts~\cite{oecd1980guidelines,gdpr2016,wp29-2013purpose} define purpose limitation as a principle, and the confinement idea goes back to Lampson~\cite{lampson1973confinement}. The closest cryptographic treatment uses attribute-based encryption to enforce purposes rather than defining a property~\cite{syedwinkler2022purpose}.

\section{Conclusion}\label{sec:conclusion}

Purpose limitation can be defined and proved as a property of a secret once two objects are kept apart: the admissible set $A(P)$ that a mechanism permits and the intended use~$I$ that the issuer declares, against which confinement, mediation, and accountability give three different guarantees. The confinement games recover existing notions in key-query form; stated against~$I$, they also register a predicate that permits too much and relate the regimes, and one payment token uses all three on a single secret. The paper introduces no new primitive or hardness assumption, and a general composition theorem is open. A mechanism enforces $A(P)$, not the human purpose, so fidelity to that purpose remains a declared, non-cryptographic obligation; within that boundary, ``this key can be used only for the operations in $A(P)$'' is a statement one can prove.

\begin{credits}
\subsubsection{\ackname}
The authors are grateful to Subhashis Banerjee, Ashoka University, for motivating the idea of purpose limitation. Thakkar is supported by the Anusandhan National Research Foundation, under the Prime Minister Early Career Research Grant ANRF/ECRG/2025/001136/ENS, and by the Vachani School of Advanced Computing, Ashoka University, under the Young Faculty Grant (00281).
\end{credits}

\bibliographystyle{splncs04}
\bibliography{references}

\appendix
\section{\texorpdfstring{Proofs of Propositions~\ref{prop:fs} and~\ref{prop:cprf}}{Proofs of Propositions 1 and 2}}\label{app:propositions}

\begin{proof}[of Proposition~\ref{prop:fs}]
Forward: given $\cA$ winning \plgame{CONF-UP}, the reduction plays the unforgeability game, receives $\vk$, sets $\pp:=\vk$, forwards each $\oracle{Constrain}(f)$ to its key oracle, and returns $\cA$'s output $(m^*,\sigma^*)$. $\cA$'s win gives $\Verify(\vk,m^*,\sigma^*)=1$ and $m^*\notin\bigcup_i\range(f_i)$, a forgery; no signing query is made. Back: a key-query forger is a \plgame{CONF-UP} adversary under the same translation. Both directions preserve the advantage.
\qed
\end{proof}

\begin{proof}[of Proposition~\ref{prop:cprf}]
The two games coincide after renaming, with the simulator that outputs the entitled leakage when it is nonempty and a uniform value otherwise. $\oracle{Constrain}(C)$ is the constrained-key oracle, and all its queries precede the challenge. If the challenge point $x^*$ is admissible for some queried~$C_i$, the entitled leakage is $F_k(x^*)$ itself and the two experiments are identical. Otherwise $C_i(x^*)=0$ for every queried~$C_i$, that is, $x^*\notin\bigcup_i I(C_i)$; the entitled leakage is empty, and real-versus-ideal is $F_k(x^*)$-versus-uniform, the constrained-PRF challenge. Advantages are equal. The proposition is for one challenge point, since a hybrid over several points needs $F_k$ at the other non-admissible points, which a game without an evaluation oracle does not supply.
\qed
\end{proof}

\section{\texorpdfstring{Proof of Theorem~\ref{thm:nonframe}}{Proof of Theorem 3}}\label{app:nonframe}

\begin{proof}
A framing makes $\Trace(\pp,(A,D),L,R)$ name an honest account~$U$, necessarily through a record verifying under $\avk_U$, for an instance $U$ never withdrew, or withdrew but neither delegated nor spent twice. Split its probability as $\varepsilon=\varepsilon_{\mathsf{forge}}+\varepsilon_{\mathsf{kept}}$ by whether $U$ signed the naming record.

\emph{Branch~1 (forged record).} A reduction guesses the framed account among the $n=\poly(\lambda)$ honest accounts, plants its EUF-CMA challenge key as that $\avk_U$, answers the account's withdrawals with its signing oracle, and outputs the naming record. Since $\rho_w$ signs $(Y,A,D)$, a record for an instance $U$ never withdrew, or with a swapped pseudonym, is such a forgery. So $\varepsilon_{\mathsf{forge}}\le n\cdot\Adv{EUF}$.

\emph{Branch~2 (genuine record).} The instance is one $U$ withdrew and neither delegated nor spent twice. $\cB$ receives $\Gamma\in\G$ and must output $\log_{g_1}\Gamma$. It guesses the framed instance among the $N=\poly(\lambda)$ honest withdrawals and plants $Y^*:=\Gamma$ as its pseudonym, without knowing $u^*=\log_{g_1}\Gamma$. Every other honest coin gets a coin identity $\cB$ samples, so $\cB$ answers issuance, spends, and $\oracle{deleg}$ for them normally. $\cB$ uses the public $(g_1,g_2)$ and never needs~$w_2$; it generates the bank keys if the bank is honest, and otherwise plays honest holders to $\cA$'s bank. It holds the honest accounts' keys, so it signs records normally.

\emph{The planted coin.} $\cB$ picks $s\gets\Zq$ and $(\sk_0,\vk_0)\gets\Kg$, sets $A:=(\Gamma g_2)^s$, and simulates $\pi$ for the true statement $Y^*=g_1^{u^*}\wedge A=(Y^*g_2)^s$ by programming the random oracle. It picks the spend values $d,r_1,r_2\gets\Zq$ in advance, sets $D:=g_1^{r_1}g_2^{r_2}A^{-d}$, signs the record on $(Y^*,A,D)$, and obtains $\sigma$ on $(A,D,P_0,\vk_0)$. When the owner spends the coin, through a handle whose chain ends at $\vk_n$, on $(M',a,t)$ with merchant nonce~$N$, $\cB$ picks a fresh $N_U$, programs $H(A,D,\vk_n,M',a,t,N,N_U):=d$, and sends the opening with $\rho:=\Sg_{\sk_n}(d)$ and that handle's chain. Since $N_U$ is fresh and uniform, the programmed point is unqueried except with negligible probability, even if a corrupt bank fixed $(A,D)$ and the transaction in advance, and (V1)--(V5) hold by construction. If the planted instance is later delegated or spent again, it is no longer a framing target and $\cB$ aborts, losing no framing on it. If the owner never spends it, the coalition must produce both openings itself.

Unless the programmed point was queried earlier or $\cB$ aborts, the simulation is perfect. In a real spend, $(r_1,r_2)=(usd+a_1,sd+a_2)$ is uniform in $\Zq^2$ with $D=g_1^{r_1}g_2^{r_2}A^{-d}$, which is the distribution $\cB$ samples; $\rho$ is a real signature, and the oracle is programmed at one high-entropy point.

\emph{Extraction.} A framing on the planted instance makes $\Trace$ recover the recorded pseudonym $Y^*=g_1^{u^*}$, which happens exactly when the linear recovery from the two openings returns~$u^*$. Openings that encode another representation of $A$ give some $u'\neq u^*$ and attribute elsewhere, so they do not frame~$U$. Hence $\cB$ outputs $u^*$, and $\Adv[\cB]{DL}\ge\varepsilon_{\mathsf{kept}}/N-\negl(\lambda)$, the negligible term collecting oracle clashes, $d=d'$ collisions, and $s=0$. With Branch~1, $\varepsilon\le n\cdot\Adv{EUF}+N\cdot\Adv{DL}+\negl(\lambda)$.
\qed
\end{proof}

\section{\texorpdfstring{Proofs of Lemma~\ref{lem:issue-view} and Theorem~\ref{thm:unlink}}{Proofs of Lemma 2 and Theorem 4}}\label{app:unlink}

\begin{proof}[of Lemma~\ref{lem:issue-view}]
Fix an accepting view $v$, which contains the bank's messages $\alpha,\beta,\gamma$, the account's $z=(Yg_2)^x$, and the holder's message~$c$. The holder's randomness is $(w,e_1,e_2)$, independent of the bank's coins, with $(e_1,e_2)$ uniform on $(\Zq\setminus\{0\})\times\Zq$. For each $w$ and $\zeta=(z',\alpha',\beta',\gamma')\in S(m)$, exactly one pair $(e_1,e_2)$ makes the holder, on $w$ against~$v$, send $c$ and output~$\zeta$.

Every output is valid and has $z'=z^s=A^x$, since $g_0^{\gamma'}=(h^c\alpha)^{e_1}g_0^{e_2}=h^{c'}\alpha'$ and $A^{\gamma'}=(z^c\beta)^{se_1}A^{e_2}=(z')^{c'}\beta'$ by the holder's checks. Conversely, let $\zeta\in S(m)$ have hash value~$c'$. Hash values are nonzero, so $c\neq0$; set $e_1:=c'/c$ and $e_2:=\gamma'-e_1\gamma$. By the checks,
\begin{align*}
  \alpha^{e_1}g_0^{e_2}&=g_0^{\gamma'}(\alpha g_0^{-\gamma})^{e_1}=g_0^{\gamma'}h^{-c'}=\alpha',\\
  \beta^{se_1}A^{e_2}&=A^{\gamma'}\bigl(\beta(Yg_2)^{-\gamma}\bigr)^{se_1}=A^{\gamma'}(z')^{-c'}=\beta',
\end{align*}
so the holder computes the same $\alpha',\beta'$, hence the same $c'$, sends $c'/e_1=c$, and outputs $\gamma'=e_1\gamma+e_2$. The relations $c=c'/e_1$ and $\gamma'=e_1\gamma+e_2$ make the pair unique. A valid signature with first component $A^x$ is $(A^x,g_0^{\tau},A^{\tau},\tau+xc')$ for $\tau=\gamma'-xc'$, so $|S(m)|=q$ for every~$m$. The bank's messages depend on its coins and on~$c$, so their probability is a factor $P(v)$ that does not depend on $(w,\zeta)$, and
\[
  \Pr[V=v,\,W=w,\,\sigma=\zeta]=P(v)\cdot\Pr[W=w]\cdot\frac{[\,\zeta\in S(m)\,]}{q(q-1)}.
\]
Summing over $\zeta$ and $w$ gives $\Pr[V=v]=P(v)/(q-1)$, and dividing gives the lemma. The hash enters only through $c'=H(m,z',\alpha',\beta')$, so adding $\vk_0$ to $m$ changes nothing, and the pads and $\sk_0$ reach the bank only through~$m$.
\qed
\end{proof}

\begin{proof}[of Theorem~\ref{thm:unlink}]
\emph{Unlinkability game.} The adversary $\cA$ plays the bank and all merchants. It runs two issuance sessions, for honest accounts of its choice, under the key for one maximal purpose~$P_0$ and names one transaction $(M',a,t)$. A challenger flips $b\gets\{0,1\}$ and redeems each coin once for $(M',a,t)$, answering $\cA$'s merchant nonces, in the order $(b,1-b)$. $\cA$ outputs $b'$ and wins if $b'=b$. Both coins share $P_0$ and $(M',a,t)$, since these are public in the transcript.

\emph{Step~1 (conditioning on issuance).} Let $V$ be $\cA$'s view of both issuance sessions, which includes the accounts $Y_1,Y_2$. Sessions are sequential and no coin is redeemed before both end, so Lemma~\ref{lem:issue-view} applies to each session in turn. Given~$V$, the coin values $W_1,W_2$ are independent with their honest distribution, and each $\sigma_j$ is uniform on $S(m_j)$, independently of everything else given~$m_j$.

\emph{Step~2 (hiding the account).} Every distribution below is conditioned on~$V$. Besides the public $(M',a,t)$ and $\cA$'s nonce~$N$, a redemption reveals $(A,D,\vk_0,\sigma,N_U,r_1,r_2,\rho)$. The nonce $N_U$ is fresh. Fix the account value~$u$; the registration check gives $Yg_2\neq1$, so $u+w_2\neq0$. Then $A=(Yg_2)^s=g_1^{s(u+w_2)}$ is uniform over $\G\setminus\{1\}$. The pads are uniform and independent of $s$ and~$\sk_0$, so $D$ is uniform, and given~$D$ the pads are uniform over the $q$ representations of~$D$. The challenge $d=H(A,D,\vk_0,M',a,t,N,N_U)$ depends on the pads only through~$D$, and $\sigma$ depends on the coin values only through $m=(A,D,\vk_0)$, so conditioning on $d$ and $\sigma$ does not change this. For fixed $(u,s,d)$ the map $(a_1,a_2)\mapsto(usd+a_1,sd+a_2)$ is a bijection from the representations of $D$ to those of $A^dD$, so, given $(A,D,d,\sigma)$, the opening is uniform over $\{(r_1,r_2):g_1^{r_1}g_2^{r_2}=A^dD\}$. The spend authorization is $\rho=\Sg_{\sk_0}(d)$, and $\sk_0$ given $\vk_0$ is distributed by $\Kg$ independently of $u$, $s$, the pads, and~$\sigma$, so $\rho$ given $(\vk_0,d)$ does not depend on the account. None of these distributions depends on~$u$ or on the session, so each redemption transcript is distributed as the same function of the public $(P_0,M',a,t)$, $\cA$'s nonce, and fresh randomness. The two coins are independent given~$V$, so $\cA$'s view, including nonces chosen after the first redemption, has the same distribution for $b=0$ and $b=1$, and the advantage is~$0$.
\qed
\end{proof}

\end{document}